\documentclass{amsart}

\usepackage{graphicx}
\usepackage[utf8]{inputenc}
\usepackage{multirow}
\usepackage{tikz}
\usetikzlibrary{shapes.geometric}
\usepackage{lscape}
\usepackage{amsmath}

\newtheorem{theorem}{Theorem}[section]

\theoremstyle{definition}

\theoremstyle{remark}

\numberwithin{equation}{section}

\title{Tetrahedal Geometry from Areas}

\author[L. Crane]{Louis Crane}
\address{Department of Mathematics, Cardwell Hall, Kansas State University, Manhattan, KS  66503}
\email{crane@math.ksu.edu}

\author[D.N.\ Yetter]{David N.\ Yetter}
\address{Department of Mathematics, Cardwell Hall, Kansas State University, Manhattan, KS  66503}
\email[Corresponding Author]{dyetter@math.ksu.edu}

\bigskip

\begin{document}
\maketitle 

\begin{abstract}
We solve a very classical problem:  providing a description of the geometry of a Euclidean tetrahedron from the initial data of the areas of the faces and the areas of the medial parallelograms of Yetter or equivalently of the pseudofaces of McConnell.  In particular, we derive expressions for the dihedral angles, face angles and (an) edge length, the remaining parts being derivable by symmetry or by identities in the classic 1902 compendium of results on tetrahedral geometry by G. Richardson.  We also provide an alternative proof using (bi)vectors of the result of Yetter that four times the sum of the squared areas of the medial parallelograms is equal to the sum of the squared areas of the faces.  Despite the classical nature of the problem, it would not have been natural to consider had it not been suggested by recent work in quantum physics.

\end{abstract}

When the first author was a graduate student, he spent a long time pondering the riddle of quantizing gravity.

Our understanding of gravity from general relativity is that it is the result of curvature in the geometry of spacetime. For this reason, the first author began contemplating the problem of the quantization of geometry.
There is one part of quantum mechanics which is a quantization of a type of geometry, and whose mathematical foundation is unusually well understood: the theory of spin.

Spin is the quantization of angular momentum. Classically, the  angular momentum of  particle is given by $\vec{r} \times \vec{p}$ the cross product of the position and momentum vectors of the particle (angular momentum is defined relative to some center). It can be thought of as the oriented area of the parallelogram generated by r and p, or double the oriented area of the triangle connecting them.

The cross product of two vectors is not so commonly thought of as a bivector because
in three dimensions it can be identified with its dual vector. This amounts to
describing an oriented area element by a normal vector of length equal to the area,
with the direction of the normal as one is taught in physics classes, using the righthand
rule (a choice of orientation on the ambient 3-dimensional space). The identification of bivectors with vectors allows us to add them using vector addition. This correctly describes the addition of angular momenta.

 It turns out however, that thinking of a spin as a quantized area element is key to the applications of spin geometry to quantum geometry.
 
 There is an important fact about spacetime which general relativity teaches us. A bounded region in spacetime can only pass a finite amount of information to its exterion. The amount of information,or more precisely the dimension of the hilbert space, is the exponential of the area of the boundary of the region in plank units.
 
 The first author decided to implement this idea by replacing the spacetime manifold with a discrete simplicial complex. Thus the non-denumerably infinite point set of classical analysis is abandoned as unphysical.
 
 A simplicial complex is a higher dimensional analog of a triangulated surface. There is one type of simplex for each dimension. A zero simplex is a point, a 1-simplex is a line segment, a 2-simplex is a triangle a 3-simplex is a tetrahedron etc.  Simplicial complexes are formed by attaching simplices along their faces.
 
 Now it turns out that the way to combine these two ideas is to construct models for quantum geometry by putting quantum spins on the two dimensional simplices of a simplicial complex. This has proved to be an interesting approach to constructing models for quantum relativity and the related but simpler problem of constructing topological quantum field theories.

In state-sum constructions of topological quantum field theories (TQFTs) and related models of gravity in 4-dimensions, spins are associated to the 2-dimensional faces of a triangulation.  %Spin is the quantum version of angular momentum.  Classically, angular momentum is represented as the cross product of two vectors, $\vec{r} \times \vec{p}$, the magnitude of which gives the area of the parallelogram spanned by the vectors, or twice the area of the triangle spanned by them, when the vectors are regarded as merely geometric quantities, rather than carrying their physical meaning.
Most models
for TQFTs and gravity in four dimensions \cite{BC98,CKY97,CY93,ELPR08} involve not only assigning a spin
to each 2-dimensional face, but another spin to each 3-simplex to give so-called (quantum)
15j-symbols (10 spins for the 2-dimensional faces, and 5 for the 3-dimensional
faces of a 4-simplex). Motivated by the suggestion that the extra spin corresponded to
(the area of) a parallelogram with vertices on the mid-points of four edges, in 1999,
the second author gave a formula for the areas of these “medial parallelogram” in terms of the six
edge lengths \cite{Y99}.  For the parallogram with vertices at the midpoints of the edges with lengths $b,c,y$ and $z$ in Figure 1 below, his formula gives the area $L$ as

\[ L = \frac{1}{8}\sqrt{4a^2x^2 - (b^2 - c^2 + y^2 - z^2)^2} \]

His observation, made soon after, that for any tetrahedron, the sum of the squares of the areas of the four faces of the tetrahedon is always four times the sum of the squares of the areas of the three medial parallelograms, was eventually published \cite{Y10} in the context of its generalization to a relation between squared hypervolumes of faces and squared hypervolumes of medial sections of simplexes, in the sense of Talata \cite{T03}, for simplexes of any dimension.

Independently in 2005, in work published only on a private website in 2012, McConnell \cite{M05} considered a notion equivalent to the medial parallelogram, which he termed ``pseudofaces,'' the quadrilaterals obtained by projecting the tetrahedron onto a plane, which (albeit described differently in his paper) may be taken without loss of generality to be the plane containing one of the medial parallelograms.  The medial parallelogram is then the Varignon parallelogram \cite{V31} of McConnell's pseudoface, its area being exactly half that of the pseudoface.  Thus the relation, discovered by McConnell, corresponding to that cited above is simply that the sum of the squares of the face areas is the sum of the squares of the pseudoface areas.

There was then the interesting question of whether the classical limit of the spin geometry was really ordinary point set geometry. This led to the
classical problem of recovering the shape of a tetrahedron from its oriented areas.
Along the way, we will see that the (bi)vector area elements of the medial parallograms
contain the same information as sums of the (bi)vector elements of the faces, and thus
contain information about the dihedral angles. Recovering the shape of the tetrahedron
reduces to finding its six edge lengths.

Aspects of this paper, in particular our new proof of the relation involving sums of
squared areas, can be thought of as “using physics to do math,” an idea popularized by Edward Witten.

It is the purpose of this note to show explicitly how any six of the seven inequivalent
areas thus associated with a tetrahedron determine the dihedral angles, face angles, and
finally the edge lengths, thus completely determining the geometry of the tetrahedron.  

We fix notation following the classic compendium of trigonometric results on tetrahedara by Richardson \cite{R02}, omitting those for which the present work has no use.  We also, for brevity, adopt his abuse of language to refer to parts of the tetrahedron by the notation for their area or length.
 
\begin{center}
\begin{figure}[h]
\setlength{\unitlength}{0.0008in}%
\begingroup\makeatletter\ifx\SetFigFont\undefined
% extract first six characters in \fmtname
\def\x#1#2#3#4#5#6#7\relax{\def\x{#1#2#3#4#5#6}}%
\expandafter\x\fmtname xxxxxx\relax \def\y{splain}%
\ifx\x\y   % LaTeX or SliTeX?
\gdef\SetFigFont#1#2#3{%
  \ifnum #1<17\tiny\else \ifnum #1<20\small\else
  \ifnum #1<24\normalsize\else \ifnum #1<29\large\else
  \ifnum #1<34\Large\else \ifnum #1<41\LARGE\else
     \huge\fi\fi\fi\fi\fi\fi
  \csname #3\endcsname}%
\else
\gdef\SetFigFont#1#2#3{\begingroup
  \count@#1\relax \ifnum 25<\count@\count@25\fi
  \def\x{\endgroup\@setsize\SetFigFont{#2pt}}%
  \expandafter\x
    \csname \romannumeral\the\count@ pt\expandafter\endcsname
    \csname @\romannumeral\the\count@ pt\endcsname
  \csname #3\endcsname}%
\fi
\fi\endgroup
\begin{picture}(4725,4671)(2701,-6088)
\thicklines
\put(3151,-5761){\line( 0, 1){  0}}
\put(3151,-5761){\line( 1, 0){4050}}
\put(7201,-5761){\line(-1, 3){1350}}
\put(5851,-1711){\line(-2,-3){2700}}
\put(5851,-1711){\line( 2,-3){1350}}
\put(7201,-3736){\line(-2,-1){4050}}
\put(7201,-3736){\line( 0,-1){2025}}
\put(4276,-3511){\makebox(0,0)[lb]{\smash{\SetFigFont{12}{14.4}{it}$a$}}}
\put(6301,-3736){\makebox(0,0)[lb]{\smash{\SetFigFont{12}{14.4}{it}$z$}}}
\put(7351,-4861){\makebox(0,0)[lb]{\smash{\SetFigFont{12}{14.4}{it}$x$}}}
\put(5101,-6061){\makebox(0,0)[lb]{\smash{\SetFigFont{12}{14.4}{it}$b$}}}
\put(6676,-2611){\makebox(0,0)[lb]{\smash{\SetFigFont{12}{14.4}{it}$y$}}}
\put(5026,-4636){\makebox(0,0)[lb]{\smash{\SetFigFont{12}{14.4}{it}$c$}}}
\put(5701,-1561){\makebox(0,0)[lb]{\smash{\SetFigFont{12}{14.4}{rm}$A$}}}
\put(7326,-3661){\makebox(0,0)[lb]{\smash{\SetFigFont{12}{14.4}{rm}$C$}}}
\put(7351,-5986){\makebox(0,0)[lb]{\smash{\SetFigFont{12}{14.4}{rm}$B$}}}
\put(2851,-5986){\makebox(0,0)[lb]{\smash{\SetFigFont{12}{14.4}{rm}$O$}}}

\end{picture}
\caption{Figure 1:  The tetrahedron $OABC$ with edge lengths labeled \label{OABC}}
\end{figure}
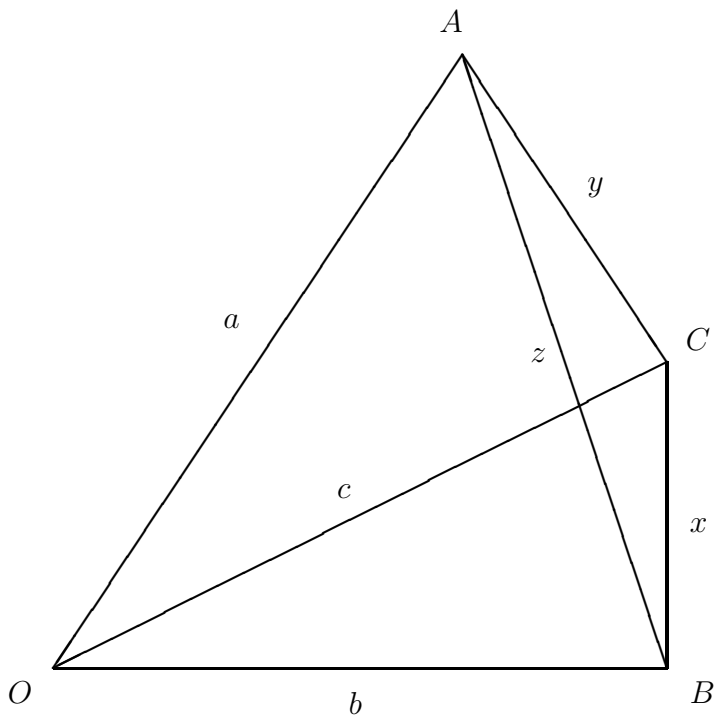
\end{center}

$OABC$ is the tetrahedron; %$V$ its volume;

$\Delta_0$, (respectively, $\Delta_1$,$\Delta_2$,$\Delta_3$) the area of the face opposite to $O$, (respectively, $A, B, C$.); 

$a, b, c, x, y, z$ the lengths of $OA, OB, OC, BC, CA,$ and $AB$, respectively;

$A$, (respectively, $B, C, X, Y, Z$ the dihedral angles whose edge is (of length) $a$ (respectively, $b, c, x, y, z$);

%$p_0$ (resp. $p_1, p_2, p_3$ the perpendiculars from $O$ (resp. $A, B, C$) to the opposite face;

%$ d_1, d_2, d_3$ the joins of the mid-points of 
%(observe that each of these is the diagonal of the two medial parallelograms); 

%$l, m, n$ the perpendicular distances between $a$ and $x$, $b$ and $y$, $c$ and $z$, respectively;

$\alpha_1$, $\beta_2$ and $\gamma_3$  the angles $BOC, COA, AOB$, respectively;

$\alpha_0$, $\beta_3$ and $\gamma_2$  the angles $BAC, OAB, CAO$, respectively;
 
$\alpha_3$, $\beta_0$ and $\gamma_1$  the angles $OBA, ABC, CBO$, respectively; and 

$\alpha_2$, $\beta_1$ and $\gamma_0$  the angles $OCA, BCO, ACB$, respectively.

As Richardson \cite{R02} observes of his notational choices,
$\alpha_i$ is not adjacent to either $a$ or $x,$ $\beta_i$ is not adjacent to either $b$ or $y$, and
$\gamma_i$ is not adjacent to either $c$ or $z$, for $i = 0,1,2, 3$,
and that the subscripts are the same as those of the faces $\Delta_i$ in which the angles lie.

Finally we need notations for the areas not considerd by Richardson \cite{R02} as they only arose in the work of the Yetter \cite{Y99, Y10} or McConnell \cite{M05}:  

Following McConnell \cite{M05}, $P, Q,$ and $R$ are the areas of the projections of the tetrahedron onto planes parallel to $a$ and $x$, $b$ and $y$, and $c$ and $z$, respectively, while the areas of the medial parallelograms lying in the plane parallel to those pairs of edges midway between them are $L, M$ and $N$, respectively.

Before turning to the derivation of dihedral angles, face angles and (an) edge length
from the areas we provide a new proof the following:

\begin{theorem}
 For any tetrahedron
\[ 4(L^2 + M^2 + N^2) = \Delta_0^2 + \Delta_1^2 + \Delta_2^2 + \Delta_3^2 \]
\end{theorem}

\begin{proof}
Regard the tetrahedron as a region in a fluid at rest, and for simplicity exerting
a constant pressure of 1. The net force acting on the region is zero, but this can be
expressed as the sum of forces acting by inward normal vectors to the faces, with
magnitude equal to the area of the face. Thus letting the addition of the vector mark to
our notation for the face areas indicate these inward normals we have
\[ \vec{\Delta_0} + \vec{\Delta_1} + \vec{\Delta_2} + \vec{\Delta_3} = \vec{0}, \]

\noindent from which it follows that the sum of any two inward normals is equal to the negative
of the sum of the other two inward normals (for example, $\vec{\Delta_0} + \vec{\Delta_1} = -\vec{\Delta_2} - \vec{\Delta_3}$).

The same argument applies to the regions bounded by a medial parallelogram, two
triangles (each a quarter of the area of the face containing it) and two trapezoids (each
three-quarters the area of the face containing it). If we let the addition of a vector mark
to the name of the area of a medial parallelogram indicate the normal vector to the
parallelogram in the direction of the non-incident edge containing $O$, we then have
\[ \vec{L} + \frac{3}{4}\vec{\Delta_3} +  \frac{3}{4}\vec{\Delta_2} +  \frac{1}{4}\vec{\Delta_1} +  \frac{1}{4}\vec{\Delta_1} = \vec{0} ,\]
\[ \vec{M} + \frac{3}{4}\vec{\Delta_3} +  \frac{3}{4}\vec{\Delta_1} +  \frac{1}{4}\vec{\Delta_2} +  \frac{1}{4}\vec{\Delta_1} = \vec{0} , \mbox{\rm and}\]
\[ \vec{N} + \frac{3}{4}\vec{\Delta_2} +  \frac{3}{4}\vec{\Delta_1} +  \frac{1}{4}\vec{\Delta_3} +  \frac{1}{4}\vec{\Delta_1} = \vec{0} .\]

\noindent Using the previous observation to rewrite the summands with coefficent $\frac{3}{4}$ or those with coefficient $\frac{1}{4}$, we have
\[ 2\vec{L} + \vec{\Delta_3} + \vec{\Delta_2} = \vec{0} = 2\vec{L} - \vec{\Delta_1} - \vec{\Delta_0}, \]
\[ 2\vec{M} + \vec{\Delta_3} + \vec{\Delta_1} = \vec{0} = 2\vec{M} - \vec{\Delta_2} - \vec{\Delta_0}, \mbox{\rm and} \]
\[ 2\vec{N} + \vec{\Delta_2} + \vec{\Delta_1} = \vec{0} = 2\vec{N} - \vec{\Delta_3} - \vec{\Delta_0}. \]

\noindent Armed with these relations, we can now consider dot products. From the tetrahedron we have that
\[ -\Delta_i^2 = -|\vec{\Delta_i}|^2 = \sum_{j\neq i} \vec{\Delta_i}\boldsymbol{\cdot} \vec{\Delta_j} \]

\noindent while from the last set of equations, we obtain
\[ 4L^2 = |\vec{\Delta_3} + \vec{\Delta_2}|^2 = |\vec{\Delta_1} + \vec{\Delta_0}|^2, \]
\[ 4M^2 = |\vec{\Delta_3} + \vec{\Delta_1}|^2 = |\vec{\Delta_2} + \vec{\Delta_0}|^2, \mbox{\rm and}\]
\[ 4N^2 = |\vec{\Delta_2} + \vec{\Delta_1}|^2 = |\vec{\Delta_3} + \vec{\Delta_0}|^2. \]

\noindent Adding all six equations expressing four times a squared parallelogram area as the square of the
magnitude of a sum of face vectors gives

\[ 8(L^2 + M^2 + N^2) = \sum_{i,j,y<j} |\vec{\Delta_i} + \vec{\Delta_j}|^2 = 
\sum_{i,j,y<j} |\vec{\Delta_i}|^2 + 2\vec{\Delta_i}\boldsymbol{\cdot} \vec{\Delta_j} + |\vec{\Delta_j}|^2 \]

Now, in the last expression, the square of each face area occurs three times. But the
middle terms involving a dot product can be rearranged to give the negatives of each
squared face area once by the relations from the tetrahedron, giving us exactly twice
the desired equation.
\end{proof}

In our notation, the basic results of McConnell \cite{M05} which motivated his definition of pseudofaces, together with the relationship between pseudofaces and medial parallelograms give
\[ \Delta_0^2 + \Delta_1^2 - 2\Delta_0\Delta_1 \cos X = P^2 = \Delta_2^2 + \Delta_3^2 - 2\Delta_2\Delta_3 \cos A = 4L^2, \]
\[ \Delta_0^2 + \Delta_2^2 - 2\Delta_0\Delta_2 \cos Y = Q^2 = \Delta_1^2 + \Delta_3^2 - 2\Delta_1\Delta_3 \cos B = 4M^2, \mbox{\rm and}\]
\[ \Delta_0^2 + \Delta_3^2 - 2\Delta_0\Delta_3 \cos Z = R^2 = \Delta_1^2 + \Delta_2^2 - 2\Delta_1\Delta_2 \cos C  = 4N^2.\]

McConnell's proof \cite{M05} of these identities is to observe that the pseudoface is the union
of two triangles sharing a common edge, and that the union of the altitudes of these,
together with the altitudes of two faces of which they are projections, when translated
into some plane perpendicular to the common edge, form a euclidean triangle, from
which the result follows by the ordinary law of cosines and the usual formula for areas
of triangles in terms of altitude and base. However, observe that when the dot product
is expressed in terms of the magnitudes of vectors and the cosine of the angle between
them, and it is recalled that dihedral angles are the angle between normal vectors to
the faces, these are exactly the six equations we just added to give the result of \cite{Y10}
relating the squares of face and parallelogram areas.

We can thus solve these equations to express the cosines of all of the six dihedral angles in terms of the seven areas.  Each depends on the areas of the faces incident with the dihedral angle and the area of the medial parallelogram which is not incident.  We give two examples and leave the rest to the interested reader to recover by symmetry:
\[ \cos A = \frac{P^2 - \Delta_2^2 - \Delta_3^2}{2\Delta_2\Delta_3} = \frac{4L^2 - \Delta_2^2 - \Delta_3^2}{2\Delta_2\Delta_3} ;\]
%\[ \cos B = \frac{Q^2 - \Delta_1^2 - \Delta_3^2}{2\Delta_1\Delta_3} = \frac{4M^2 - \Delta_1^2 - \Delta_3^2}{2\Delta_1\Delta_3} \]
%\[ \cos C = \frac{R^2 - \Delta_1^2 - \Delta_2^2}{2\Delta_1\Delta_2} = \frac{4N^2 - \Delta_1^2 - \Delta_2^2}{2\Delta_1\Delta_2} \]
\[ \cos X = \frac{P^2 - \Delta_0^2 - \Delta_1^2}{2\Delta_0\Delta_1} = \frac{4L^2 - \Delta_0^2 - \Delta_1^2}{2\Delta_0\Delta_1} . \]
%\[ \cos Y = \frac{Q^2 - \Delta_0^2 - \Delta_2^2}{2\Delta_0\Delta_2} = \frac{4M^2 - \Delta_0^2 - \Delta_2^2}{2\Delta_0\Delta_2} \]
%\[ \cos Z = \frac{R^2 - \Delta_0^2 - \Delta_3^2}{2\Delta_0\Delta_3} = \frac{4N^2 - \Delta_0^2 - \Delta_3^2}{2\Delta_0\Delta_3} \]

From this and the classical result that the dihedral angles on edges meeting at any vertex are the angles of a spherical triangle with the face angles at the vertex as the central angles, we can use the spherical law of cosines to express the angles between the edges in terms of the areas.  We give the face angles for the face opposite $O$, leaving the others to the interested reader.
\[ \cos \alpha_0 = \frac{\cos A + \cos Y \cos Z}{\sin Y \sin Z}. \]
\[ \cos \beta_0 = \frac{\cos B + \cos X \cos Z}{\sin X \sin Z}. \]
\[ \cos \gamma_0 = \frac{\cos C + \cos X \cos Y}{\sin X \sin Y}. \]

One can substitute the previous expressions for the dihedral angles in terms of areas to obtain formulas directly expressing the face angles in terms of the areas.  

\[ \cos \alpha_0 = \sqrt{\frac{[2\Delta_0^2 P^2 - 2\Delta_0^2 \Delta_2^2 - 2\Delta_0^2 \Delta_3^2 + (Q^2 - \Delta_0^2 - \Delta_2^2)(R^2 - \Delta_0^2 - \Delta_3^2)]^2}{[4\Delta_0^2 \Delta_2^2 - (Q^2 - \Delta_0^2 - \Delta_2^2)^2][4\Delta_0^2 \Delta_3^2 - (R^2 - \Delta_0^2 - \Delta_3^2)^2]}} . \]

\[ \cos \beta_0 = \sqrt{\frac{[2\Delta_0^2 Q^2 - 2\Delta_0^2 \Delta_1^2 - 2\Delta_0^2 \Delta_3^2 + (P^2 - \Delta_0^2 - \Delta_1^2)(R^2 - \Delta_0^2 - \Delta_3^2)]^2}{[4\Delta_0^2 \Delta_1^2 - (P^2 - \Delta_0^2 - \Delta_1^2)^2][4\Delta_0^2 \Delta_3^2 - (R^2 - \Delta_0^2 - \Delta_3^2)^2]}} .\]

\[ \cos \gamma_0 = \sqrt{\frac{[2\Delta_0^2 R^2 - 2\Delta_0^2 \Delta_1^2 - 2\Delta_0^2 \Delta_2^2 + (P^2 - \Delta_0^2 - \Delta_1^2)(Q^2 - \Delta_0^2 - \Delta_2^2)]^2}{[4\Delta_0^2 \Delta_1^2 - (P^2 - \Delta_0^2 - \Delta_1^2)^2][4\Delta_0^2 \Delta_2^2 - (Q^2 - \Delta_0^2 - \Delta_2^2)^2]}} . \]

And similarly for the sines:

\[ \sin \alpha_0 = \sqrt{1 - \frac{[2\Delta_0^2 P^2 - 2\Delta_0^2 \Delta_2^2 - 2\Delta_0^2 \Delta_3^2 + (Q^2 - \Delta_0^2 - \Delta_2^2)(R^2 - \Delta_0^2 - \Delta_3^2)]^2}{[4\Delta_0^2 \Delta_2^2 - (Q^2 - \Delta_0^2 - \Delta_2^2)^2][4\Delta_0^2 \Delta_3^2 - (R^2 - \Delta_0^2 - \Delta_3^2)^2]}} ;\]

\[ \sin \beta_0 = \sqrt{1 - \frac{[2\Delta_0^2 Q^2 - 2\Delta_0^2 \Delta_1^2 - 2\Delta_0^2 \Delta_3^2 + (P^2 - \Delta_0^2 - \Delta_1^2)(R^2 - \Delta_0^2 - \Delta_3^2)]^2}{[4\Delta_0^2 \Delta_1^2 - (P^2 - \Delta_0^2 - \Delta_1^2)^2][4\Delta_0^2 \Delta_3^2 - (R^2 - \Delta_0^2 - \Delta_3^2)^2]}} ;\]

\[ \sin \gamma_0 = \sqrt{1 - \frac{[2\Delta_0^2 R^2 - 2\Delta_0^2 \Delta_1^2 - 2\Delta_0^2 \Delta_2^2 + (P^2 - \Delta_0^2 - \Delta_1^2)(Q^2 - \Delta_0^2 - \Delta_2^2)]^2}{[4\Delta_0^2 \Delta_1^2 - (P^2 - \Delta_0^2 - \Delta_1^2)^2][4\Delta_0^2 \Delta_2^2 - (Q^2 - \Delta_0^2 - \Delta_2^2)^2]}} .\]

\medskip

Having thus determined the angles between the edges, we can now use the ordinary law of sines and the area formula for a triangle in side-angle-side form to express the edge lengths in terms of the areas:

Consider the triangle opposite O, with area $\Delta_0$

The area is given by
\[ \Delta_0 = \frac{1}{2}\sin \alpha_0 y z.\]

\medskip

Now by the law of sines, we have $y = x\frac{\sin \beta_0}{\sin \alpha_0}$ and $z = x\frac{\sin \gamma_0}{\sin \alpha_0}$. So we may rewrite the area as
\[ \Delta_0 = \frac{1}{2} \frac{\sin \beta_0 \sin \gamma_0}{\sin \alpha_0} x^2. \]

From which we obtain
\[ x = \sqrt{\frac{ 2\Delta_0 \sin \alpha_0}{\sin \beta_0 \sin \gamma_0}}. \]

Substituting the expressions above for the sines, then expresses the length $x$ directly in terms of the seven areas, any one of which can be determined from the other six.  Unfortunately the resulting expression, even after all readily apparent simplifictions have been performed, is too large to conveniently typeset.  As a practical matter computing the sines and substituting the results into the last expression would be the convenient way to find lengths from areas.

\end{document}